\documentclass[10pt,english,conference]{IEEEtran}
\usepackage[T1]{fontenc}
\usepackage[latin1]{inputenc}
\usepackage{geometry}
\usepackage{textcomp}
\usepackage{slashed}
\usepackage{color}
\usepackage{amsthm}
\usepackage{amsmath}
\usepackage{amssymb}
\usepackage{graphicx}
\usepackage{esint}

\makeatletter

\newcommand{\Ex}{\mathbb{E}}
\renewcommand{\Pr}{\mathbb{P}}

\newcommand{\vect}[1]{\boldsymbol{#1}}

\newcommand{\set}[1]{\mathcal{#1}}

\newtheorem{lemma}{Lemma}
\newtheorem{theorem}{Theorem}
\newtheorem{corollary}{Corollary}
\newtheorem{remark}{Remark}
\newtheorem{definition}{Definition}
\newtheorem{fact}{Fact}

\usepackage{babel}

\makeatother

\author{\IEEEauthorblockN{Farzin~Haddadpour\IEEEauthorrefmark{1}, Mahdi~Jafari~Siavoshani\IEEEauthorrefmark{2}, Mayank Bakshi\IEEEauthorrefmark{2}, Sidharth Jaggi\IEEEauthorrefmark{2}}
\IEEEauthorblockA{\IEEEauthorrefmark{1}School of Electrical Engineering\\
Sharif University of Technology, Tehran, Iran\\
Email: \texttt{farzinhaddadpour@gmail.com}}
\IEEEauthorblockA{\IEEEauthorrefmark{2}Institute of Network Coding\\
Chinese University of Hong Kong, Hong Kong\\
Email: \texttt{\{mahdi,mayank\}@inc.cuhk.edu.hk}, \texttt{jaggi@ie.cuhk.edu.hk}}}

\begin{document}
\title{On AVCs with Quadratic Constraints}

\maketitle

\begin{abstract}
In this work we study an Arbitrarily Varying Channel
(AVC) with quadratic power constraints on the transmitter and a so-called ``oblivious'' jammer (along with additional AWGN)
under a {\it maximum probability of error} criterion, and no private randomness between the transmitter and the receiver. This is in contrast to similar AVC models under the {\it average probability of error} criterion considered in~\cite{CsisNar-IT91}, and models wherein common randomness is allowed~\cite{NarHugh-IT87} -- these distinctions are important in some communication scenarios outlined below.

We consider the regime where the jammer's power constraint is smaller than the transmitter's power constraint (in the other regime it is known no positive rate is possible). For this regime we show the existence of stochastic codes (with {\it no common randomness} between the transmitter and receiver) that enables reliable communication at the same rate as when the jammer is replaced with AWGN with the same power constraint. This matches known information-theoretic outer bounds. In addition to being a stronger result than that in~\cite{CsisNar-IT91} (enabling recovery of the results therein), our proof techniques are also somewhat more direct, and hence may be of independent interest.

%
\end{abstract}

\section{Introduction}\label{sec:Introduction}

Aerial Alice is flying in a surveillance plane high over Hostile Harry's territory. She wishes to relay her observations of Harry's troop movements back to Base-station Bob over $n$ channel uses of an AWGN channel with variance $\sigma^2$. Harry obviously wishes to jam Alice's transmissions. However, both Alice's transmission energy and Harry's jamming energy are constrained -- they have access to energy sources of $nP$ and $n\Lambda$ Joules respectively.\footnote{These are so-called {\it peak power} constraints -- they must hold for {\it all} codewords, rather than averaged over all codewords {\it average power} constraints. If the peak power constraints are relaxed to average power constraints, for either Alice's transmissions, or Harry's jamming (or both), it is known~\cite{NarHugh-IT87} that standard capacity results do not hold -- only ``$\lambda$-capacities'' exist.} Harry {already} knows {\it what} {\it message} Alice wants to transmit (after all, he knows the movements of his own troops), and also {\it roughly how} she'll transmit it ({\it i.e.}, her {\it communication protocol/code}, having recently captured another surveillance drone)
but he doesn't know {\it exactly how} she'll transmit it ({\it i.e.}, her {\it codeword} -- for instance, Alice could choose to focus her transmit power on some random subset of the $n$ channel uses). Further, since Alice's transmissions are very quick, Harry has no time to tune his jamming strategy to Alice's actual codeword -- he can only jam based on his prior knowledge of Alice's code, and her message.\footnote{Alternatively, Alice could split her energy budget to concurrently transmit one symbol on $n$ different frequencies -- these together could comprise her codeword. Given such a strategy, since Harry doesn't know Alice's codeword, he is unable to make his jamming strategy depend explicitly on the codeword Alice actually transmits.}

Even in such an adverse jamming setting we demonstrate that Alice can communicate with Bob at a rate equalling $\frac{1}{2}\log \left (1+\frac{P}{\Lambda+\sigma^{2}} \right )$ as long as $P>\Lambda$. Note that this equals the capacity of an AWGN with noise parameter equal to $\Lambda+\sigma^{2}$ -- this means that no ``smarter'' jamming strategy exists for Harry than simply behaving like AWGN with variance $\Lambda$. If $P< \Lambda$ no positive rate is possible since Harry can ``spoof'' by transmitting a fake message using the same strategy as Alice -- Bob is unable to distinguish between the real and fake transmissions\footnote{Such a jamming strategy is equivalent to the more general {\it symmetrizability} condition in the AVC literature (see, for instance~\cite{CsiN:88, CsiN:88a}, and \cite{LapNar-IT98}).}.

\subsection{Relationship with prior work}
The model considered in this work is essentially a special type of Arbitrarily Varying Channel (AVC) for which, to the best of our knowledge, the capacity has not been characterized before in the literature.
The notion of AVCs was first introduced by Blackwell \emph{et al.}~\cite{BlaBT:59,BlBrTho-AnMatStat60}, to capture communication models wherein channel have unknown parameters that may
vary arbitrarily during the transmission of a codeword. The case when both the transmitter and the jammer operate under constraints (analogous to the quadratic constraints in this work) has also been considered~\cite{CsiN:88, CsiN:88a}.
For an 
extensive survey on AVCs the reader may refer to the excellent survey~\cite{LapNar-IT98} 
and the references therein.

The class of AVCs over discrete alphabets has been studied in great detail in the 
literature~\cite{LapNar-IT98}. However, less is known about AVCs with continuous alphabets. The bulk of the work on continuous alphabet AVCs (outlined below in this section) focuses on quadratically-constrained AVCs. This is also the focus of our work.

It is important to stress several features of the model considered in this work, and the differences with prior work:
\begin{itemize}
\item {\underline {\it Stochastic encoding:}} To generate her codeword from her message, Alice is allowed to use private randomness (known only to her {\it a priori}, but not to Harry {\it or} Bob.
This is in contrast to the {\it deterministic encoding} strategies often considered in the information theory/coding theory literature, wherein the codeword is a deterministic function of the message.
\item {\underline {\it Public code:}} Everything Bob knows about Alice's transmission {\it a priori}, Harry also knows.\footnote{This requirement is an analogue for communication of  Kerckhoffs' Principle~\cite{Ker:83} in cryptography, which states that in a secure system, everything about the system is public knowledge, except possibly Alice's {\it private} randomness.} This is in contrast to the {\it randomized encoding} model also considered in the literature (see for instance~\cite{NarHugh-IT87, AgaSM:06}), in which it is critical that Alice and Bob share {\it common randomness} that is unknown to Harry.
\item {\underline {\it Message-aware jamming:}} The jammer is already aware of Alice's message. This is one important difference in our model, from the model in the work closest to ours, that of~\cite{CsisNar-IT91}.
\item {\underline {\it Oblivious adversary:} The jammer has no extra knowledge of the codeword being transmitted than what he has already gleaned from his knowledge of Alice's code and her message. This is in contrast to the {\it omniscient adversary} often considered in the coding theory literature.}
\end{itemize}

These model assumptions are equivalent to requiring public 
stochastic codes with small maximum error of probability against 
an oblivious adversary. Several papers also operate under {\it some} 
of these assumptions, but as far as we know, none examines 
the scenario where {\it all} these constraints are active.

The literature on {\it sphere packing} focuses on an AVC model
 wherein zero-error probability of decoding is required (or, equivalently, 
 when the probability (over Alice's codeword and Harry's jamming actions) 
 of Bob's decoding error is required to equal zero). Inner and 
 outer bounds were obtained by Blachman~\cite{Bla:62, Bla:62a}.
Like several other zero-error communication problems (including 
Shannon's classic work~\cite{Sha:56}) characterization of the 
optimal throughput possible is challenging, and in general still 
an open problem.\footnote{The literature on 
{\it Spherical Codes} (see~\cite{Wyn:67},~\cite{HamZ:97},~and~\cite{HamZ:97a} 
for some relatively recent work) looks at the related problem of 
packing unit hyperspheres on the {\it surface} of a hypershere. 
This corresponds to design of codes where each codeword meets 
the quadratic power constraint with equality, rather than allowing 
for an inequality.}

Other related models include: 
\begin{itemize}
\item The {\it vector Gaussian AVC}~\cite{HugN:88a}. As in the ``usual'' vector Gaussian channels, optimal code designs require ``waterfilling''.
\item The {\it per-sequence/universal} coding schemes in~\cite{LomF:11}. 
\item The {\it correlated/myopic} jammers in~\cite{Med:97,Sar:12}, wherein jammers obtain a noisy version of Alice's transmission and base their jamming strategy on this. 
\item The {\it joint source-channel coding, and coding with feedback} models considered by Ba\c{s}ar~\cite{Bas:83, Bas:89}.
\item Several other AVC variants, including {\it dirty paper coding}, in~\cite{SarG:12}.
\end{itemize}

We summarize some of the results mentioned above in Table~\ref{tab:ComparisonAVC}.


\begin{table*}
\centering
\begin{tabular}{|c|c|c|}
\hline
 & \textbf{Error Criterion} & \textbf{Capacity}\\
\hline
Blachman \cite{Bla:62} & $\sup_{\vect{s}} \sup_{\psi} \left[ \phi(\psi(i)+\vect{s}+\vect{V}) \neq i \right] \leq \epsilon$ & upper and lower bounds for the capacity\\
\hline
Hughes \& Narayan \cite{NarHugh-IT87} &  $\sup_{\vect{s}} \max_{i} \Pr_{(\Phi,\Psi),V} \left[ \Phi(\Psi(i) + \vect{s} + \vect{V}) \neq i\right] \le \epsilon$ &  $C=\frac{1}{2}\log(1+\frac{P}{\Lambda+\sigma^{2}})$\\
\hline
Csiszar \& Narayan \cite{CsisNar-IT91} &  $\sup_{\vect{s}} \frac{1}{M} \sum_{i=1}^M \Pr_{V} \left[ \phi(\psi(i) + \vect{s} +\vect{V}) \neq i \right] \le \epsilon$  & $C=\left\{ \begin{array}{ll}
\frac{1}{2}\log(1+\frac{P}{\Lambda+\sigma^{2}}) & \text{if} \: P>\Lambda\\
0 & \text{Otherwise}.
\end{array}\right.$\\
\hline
Our Setup & $\sup_{\vect{s}} \max_{i} \Pr_{\Psi,V} \left[ \phi(\Psi(i) + \vect{s} + \vect{V}) \neq i \right] \leq \epsilon$ &  $C=\left\{ \begin{array}{ll}
\frac{1}{2}\log(1+\frac{P}{\Lambda+\sigma^{2}}) & \text{if} \: P>\Lambda\\
0 & \text{Otherwise}.
\end{array}\right.$\\
\hline
\end{tabular}
\caption{Comparison of existing results on Quadratic-constrained AVCs with AWGN.}
\label{tab:ComparisonAVC}
\end{table*}


\section{Notation and Problem Statement}\label{sec:ProbStatementAVC}

\subsection{Notation}
Throughout the paper, we use capital letters to denote random variables 
and random vectors, and corresponding lower-case letters to 
denote their realizations. Moreover, bold letters
are reserved for vectors and calligraphic symbols denote sets. Random sets
are represented by an extra star as superscripts. 
Some constants are also denoted by capital letters.
Our convention is summarized in Table~\ref{tab:NotationConvention}.
\begin{table}
\centering
\begin{tabular}[h]{|c|c|c|c|}
\hline
 & Deterministic & Random & Realization\\
\hline
Scalar & $a,k,N$ & $X,Y$ & $x,y$\\
\hline
Vector & $\vect{v}$ & $\vect{V}$ & $\vect{v}$\\
\hline
Set & $\mathcal{C}$ & $\mathcal{C}^*$ & $\mathcal{C}$ \\
\hline
\end{tabular}
\caption{Examples of our notation convention for different variables.}
\label{tab:NotationConvention}
\end{table}

We use $N(a,\sigma^2)$ to denote for a Gaussian random variable
with mean $a$ and variance $\sigma^2$. To denote a ball in an $n$-dimensional 
real space of radius $r$ which centered at the point $\vect{c}\in\mathbb{R}^n$, 
we write $\mathcal{B}_n(\vect{c},r)$. In Table~\ref{tab:NotationSummary}, we summarize
the notation used in this paper.
\begin{table}
\centering
\begin{tabular}[h]{|l|p{195pt}|}
\hline
 \textbf{Symbol} & \textbf{Meaning}\\
\hline
$\Psi(i)$ & Stochastic encoder applied to the message $i$\\
\hline
$\phi(Y)$ & Deterministic decoder\\
\hline
$e(\vect{s},i)$ & Error probability (over the stochastic encoder and the channel noise) for a fixed message $i$ and jamming vector $\vect{s}$\\
\hline
$e_{\max}(\vect{s})$ & Maximum (over messages) error probability  for a fixed jamming vector $\vect{s}$\\
\hline
$N(a,\sigma^2)$ & Gaussian random variable with mean $a$ and variance $\sigma^2$\\
\hline
$\mathcal{B}_n(\vect{c},r)$ & A ball of radius $r$ in $\mathbb{R}^n$ which centered at $\vect{c}\in\mathbb{R}^n$\\
\hline
\end{tabular}
\caption{Commonly used symbols.}
\label{tab:NotationSummary}
\end{table}

\subsection{Problem Statement}\label{sec:ProbStatement-ScalarAVC}
In this paper we study the capacity of a quadratic constrained AVC
with stochastic encoder under the attack of a malicious adversary
who knows the transmitted message but is oblivious to the
actual transmitted codewords. 

Let the input and output of the channel are denoted by the random 
variables $X$ and $Y$ where $X,Y\in\mathbb{R}$. 
Then, formally, the channel is defined as follows
\begin{equation}\label{eq:ChannelSetup}
Y=X+S+V,
\end{equation}
where $S\in\mathbb{R}$ is the channel state chosen by a malicious 
adversary and $V\sim N(0,\sigma^2)$ is Gaussian random variable. 
Here we assume that the noise $V$ is independent
over different uses of channel \eqref{eq:ChannelSetup}.
The channel input is subjected to a peak power constraint 
as follows
\begin{equation}\label{eq:InputConstraint-AVC}
	\| \vect{x} \|^2 = \sum_{i=1}^n x_i^2 \le nP,
\end{equation}
and the permissible state sequences are those satisfying
\begin{equation}\label{eq:StateConstraint-AVC}
	\| \vect{s} \|^2 = \sum_{i=1}^n s_i^2 \le n\Lambda.
\end{equation}
The problem setup is depicted pictorially in Figure~\ref{fig:ProblemSetup-NoisyCase}.

\begin{figure}
\centering
\includegraphics[scale=0.53]{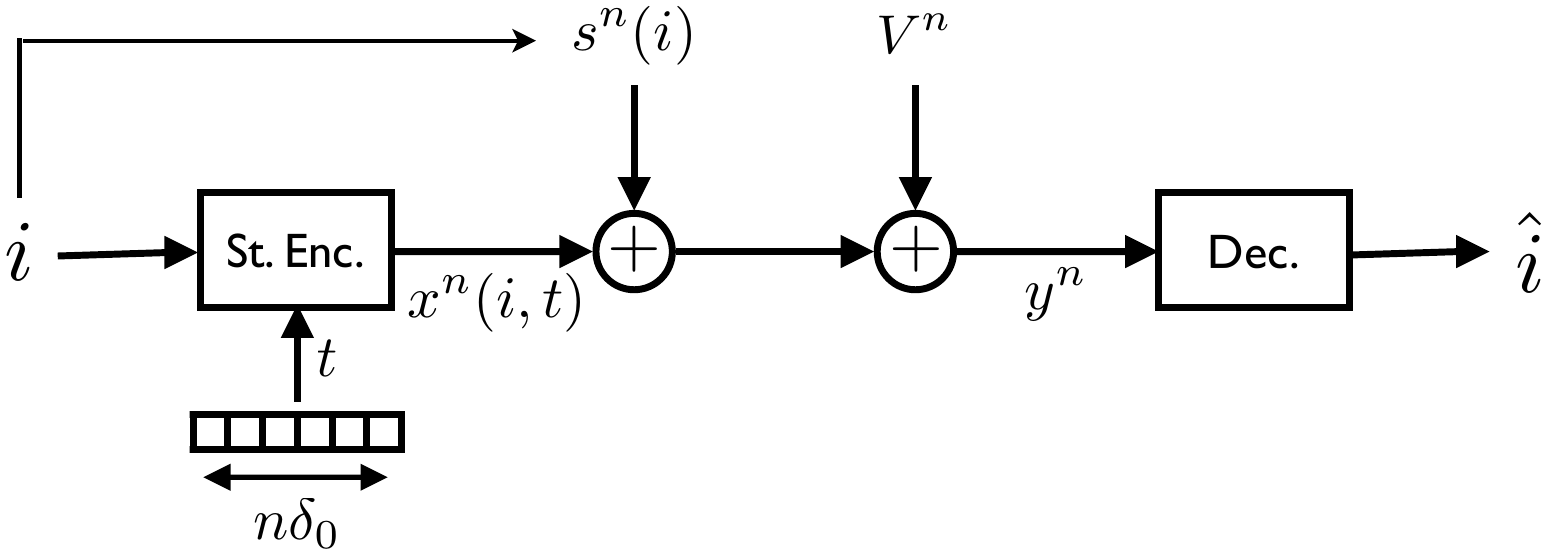}
\vspace{-7pt}
\caption{A power constraint AVC with stochastic encoder. Here 
we assume that the adversary has access to the transmitted message $i$
but not to the transmitted codeword $x^n(i,t)$.}
\label{fig:ProblemSetup-NoisyCase}
\end{figure}

A \emph{code with stochastic encoder} $(\Psi,\phi)$ of block-length $n$ 
consists of a set of encoders that are denoted by a random variable
$\Psi:\  \{1,\ldots, M \}\mapsto \mathbb{R}^n$ and a deterministic decoder 
$\phi:\ \mathbb{R}^n\mapsto \{0,\ldots,M\}$ where $0$ denote for 
an error and $M\triangleq e^{nR}$ is the number of messages\footnote{For 
notational convenience we assume that $e^{nR}$ is an integer.}. 
Each encoder $\psi$ is constructed by a set of codewords 
$\{ \vect{x}_1,\ldots,\vect{x}_M \}$ from $\mathbb{R}^n$.

Here in this paper, we focus on the \emph{maximum
probability of error}. First, for a fixed jamming vector $\vect{s}$, 
let us define the probability of error
given that the message $i$ has been sent as follows
\begin{equation}\label{eq:ErrorDef_GeneralCase}
e(\vect{s}, i) \triangleq \Pr_{\Psi,V} \left[ \phi \left( \Psi(i)+\vect{s} + \vect{V} \right) \neq i\right].
\end{equation}
Then the maximum probability of error for a fixed $\vect{s}$
is defined by
\begin{equation}\label{eq:ErrorMaxDef_GeneralCase}
e_{\mathrm{max}} (\vect{s}) \triangleq \max_{i\in\{1,\ldots, M\}} e(\vect{s} , i).
\end{equation}
Now the \emph{capacity} for the above channel can be stated as in 
Definition~\ref{def:Capacity}.

\begin{definition}\label{def:Capacity}
The capacity $C$ of an AVC with stochastic encoder
under the quadratic transmit constraint $P$ and jamming constraint $\Lambda$
is the supremum over the set of real numbers
such that for every $\delta>0$ and sufficiently large 
$n$ there exist codes  with stochastic encoder $(\Psi,\phi)$ that
satisfies the following conditions.
First, for the number of messages $M$ encoded by the code
we have $M>\exp(n(C-\delta))$. Moreover, each codeword satisfies the quadratic
constraint~\eqref{eq:InputConstraint-AVC} and finally for the code we have
\begin{equation*}
\lim_{n\to\infty} \sup_{\vect{s}: \|\vect{s} \|^2 \le n\Lambda} 
	e_{\mathrm{max}}(\vect{s}) = 0.
\end{equation*}
\end{definition}

\section{Main Results}\label{sec:MainResultAVC}
The main results of the paper, stated in Theorem~\ref{thm:CapacityAVC-Noiseless}
and its corollary.

\begin{theorem}\label{thm:CapacityAVC-Noisy}
The capacity of a quadratic-constrained AVC channel under the maximum 
probability of error criterion with transmit constraint $P$ and
jamming constraint $\Lambda$ and additive Gaussian noise
of power $\sigma^2$ is given by
\[
C=\left\{ \begin{array}{ll}
\frac{1}{2}\log(1+\frac{P}{\Lambda+\sigma^{2}}) & \mathrm{if} \; P>\Lambda,\\
0 & \mathrm{Otherwise}.
\end{array}\right.
\]
\end{theorem}

\begin{remark}
The result of Theorem~\ref{thm:CapacityAVC-Noisy} matches the result 
of stochastic encoder over discrete alphabets \cite{Ahls:ZWVG78}, \cite[Theorem~7]{LapNar-IT98}, 
in which it is shown that for the \emph{average}
probability of error criterion, using a stochastic 
encoder doesn't increase the capacity.
Because the number of possible adversarial actions here is uncountably 
large, the technique of \cite{Ahls:ZWVG78}, which relies on taking a union 
bound over at most exponential-sized set of possible adversarial actions, 
does not work.
\end{remark}

\begin{corollary}\label{thm:CapacityAVC-Noiseless}
The capacity of a quadratic-constrained AVC  under the maximum 
probability of error criterion with transmit constraint $P$ and
jamming constraint $\Lambda$ is given by
\[
C=\left\{ \begin{array}{ll}
\frac{1}{2}\log(1+\frac{P}{\Lambda}) & \mathrm{if} \;  P>\Lambda,\\
0 & \mathrm{Otherwise}.
\end{array}\right.
\]
\end{corollary}

\section{Proof of Main Results}\label{sec:ProofMainResultAVC}
In this section, we present the proof of 
Theorem~\ref{thm:CapacityAVC-Noisy} and its corollary. 
The proof of the converse
parts of Theorem~\ref{thm:CapacityAVC-Noisy} is stated in Section~\ref{sec:AVC-Noisy-Proof-Converse}. 

For the achievability part of Theorem~\ref{thm:CapacityAVC-Noisy},
we claim that the same \emph{minimum distance decoder} proposed in
\cite{CsisNar-IT91} to achieve the capacity for the average probability 
of error criterion, which is given by 
\begin{equation}\label{eq:MinDistncDecoder-AVC}
\phi(\vect{y})=\left\{ \begin{array}{ll}
i &  \mathrm{if} \; \|\vect{y}-\vect{x}_{i}\|^{2} < \|\vect{y}-\vect{x}_{j}\|^{2},\quad \mathrm{for} \; j\neq i,\\
0 &  \mathrm{if} \; \mathrm{no\; such} \; i:1\leq i\leq M\; \mathrm{exists},
\end{array}\right.
\end{equation}
also achieves the capacity for the maximum 
probability of error criterion.

Note that in order to show the suprimum over $\vect{s}$ subject to 
\eqref{eq:StateConstraint-AVC} of $e_{\text{max}}(\vect{s})$
goes to zero it is sufficient to show that for every message $i$
the suprimum over $\vect{s}$ subject to \eqref{eq:StateConstraint-AVC}
of $e(\vect{s}, i)$ goes to zero.

To communicate, Alice (the transmitter) randomly picks a codebook $\set{C}$ and fixes it.
The codebook $\set{C}$ comprises 
$e^{n(\delta_0+R)}$ codewords $\vect{x}(i,t)$, $1\le i\le e^{nR}$ and
$1\le t\le e^{n\delta_0}$, each chosen uniformly at random and 
independently from a sphere of radius $\sqrt{nP}$
as it is shown in Figure~\ref{fig:Stochastic_Codebook_AVC} (caption (a)).
Then, the $i$th row of the codebook, i.e., 
$\{ \vect{x}(i,1),\ldots,\vect{x}(i,e^{n\delta_0}) \}$, is assigned 
to the $i$th message. In order to transmit
the message $i$, the encoder randomly picks a codeword from the
$i$th row of the codebook and sends it over the channel.
\begin{figure}
\centering
\includegraphics[scale=1.6]{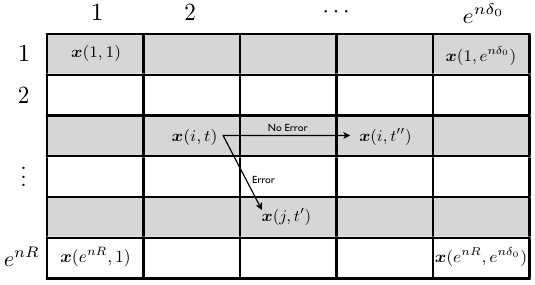}
\caption{(a) The codebook is constructed such that for sending a message 
$i\in \{1,\ldots,e^{nR}\}$ the encoder chooses one of the $e^{n\delta_0}$
codewords randomly from the $i$th row of the above table.
(b) Assuming that the codeword $\vect{x}(i,t)$ is sent, in our
model an error occurs if the ML decoder declares $\vect{x}(j,t')$
for some $j\neq i$. Note that there is no error if the decoder declares another
codeword from the $i$th row.}
\label{fig:Stochastic_Codebook_AVC}
\end{figure}

Now, given that the message $i$ has been transmitted, the error 
probability $e(\vect{s}, i)$ of an stochastic code 
used over a quadratic-constrained  AVC under the use of the minimum 
distance decoder (defined by \eqref{eq:MinDistncDecoder-AVC})
equals
\begin{align}\label{eq:Emax_given_i_noisy_AVC}
e( \vect{s} , i ) =& \Pr_{\Psi,V} \left[ \phi \left( \Psi(i)+\vect{s} + \vect{V} \right) \neq i\right] \nonumber\\
=& \Pr_{T} \Pr_{V} \Big[ \| \vect{x}(i,T)+\vect{s}+\vect{V}-\vect{x}(j,t') \|^{2} \nonumber\\
&\quad\leq \| \vect{s}+\vect{V} \|^{2} \; \text{for some $i\neq j$ and $t'$} \Big] \nonumber\\
\stackrel{}{=} & \Pr_T {\Pr}_{V} \Big[ \langle \vect{x}(j,t'),\vect{x}(i,T) + \vect{s} + \vect{V} \rangle \ge nP \nonumber\\
&\quad +\langle \vect{x}(i,T), \vect{s} + \vect{V}\rangle\; \text{for some $j\neq i$ and $t'$} \Big].
\end{align}
where $T$ is a uniformly distributed random variable
defined over the set $\{1,\ldots,e^{n\delta_0}\}$.
Figure~\ref{fig:Stochastic_Codebook_AVC} (caption (b)) pictorially 
demonstrates the decoding errors at the decoder.
\subsection{Achievability proof of Theorem~\ref{thm:CapacityAVC-Noisy}}

The main step in proving the achievability part of Theorem~\ref{thm:CapacityAVC-Noisy}
consists in asserting the doubly exponential probability bounds which
is stated in Lemma~\ref{lem:CodeProp_DoublyExpProb_Noisy}.
\begin{lemma}\label{lem:CodeProp_DoublyExpProb_Noisy}
Let $\set{C}^*=\{\vect{X}(i, t)\}$ in which $1\leq i\leq \exp(nR)$ and $1\leq t\leq \exp(n\delta_{0})$
be a random codebook comprises of independent random vectors $\vect{X}(i,t)$ each uniformly distributed
on the $n$-dimensional sphere of radius $\sqrt{nP}$. First, fix a vector 
$\vect{s}\in\mathcal{B}_n( 0,\sqrt{n\Lambda} )$. Then for every 
$\delta_0>\delta_1>0$ and for sufficiently large $n$ if
$R<\frac{1}{2}\log\left(1+\frac{P}{\sigma^2+\Lambda} \right)$ we have
\begin{align*}
&{\Pr}_{\set{C}^*} \bigg[ \Pr_T  \Pr_{V} \Big[ \langle \vect{X}(j,t'),\vect{X}(i,T) +\vect{s}+\vect{V}\rangle \geq nP \\
&\quad\quad +\langle \vect{X}(i,T), \vect{s} +\vect{V} \rangle \; \mathrm{for}\; \mathrm{some}\; \mathrm{j}\neq \mathrm{i}\; \mathrm{and}\; t'  \Big] \geq K e^{-n\delta_{1}} \bigg]\\
&\leq \exp \Big( -(K\log{2}-10) \exp((\delta_{0}-\delta_{1})n) \Big).
\end{align*}
\end{lemma}
\begin{proof}
For the proof refer to the appendix.
\end{proof}

\begin{lemma}[Quantizing Adversarial Vector] \label{lem:SmallPerturbation_s-NoisyAVC}
For a fixed jamming vector $\vect{s}$, for sufficiently 
small $\varepsilon>0$, and for every $\delta_0>\delta_1>0$, 
there exists a codebook $\set{C}=\{\vect{x}(i,t)\}$ of rate 
$R\le \frac{1}{2}\log(1+\frac{P}{\Lambda+\sigma^2})$ comprises of 
vectors $\vect{x}(i, t) \in\mathbb{R}^{n}$ of size $\sqrt{nP}$ with 
$1\leq i\leq e^{nR}$ and $1\leq t\leq e^{n\delta_{0}}$
which performs well over the AVC defined in 
Section~\ref{sec:ProbStatementAVC} for all 
$\vect{s}'\in\mathcal{B}_n(\vect{s},\varepsilon)$, i.e., 
it satisfies 
\begin{align}\label{eq:Lem_QuantAdvrsryVec-Noisy-ErrEq}
e(\vect{s}, i) &= {\Pr}_{T}  \Pr_V \Big[ \langle \vect{x}(j,t'),\vect{x}(i,T)+\vect{s}+\vect{V} \rangle \nonumber\\
&\hspace{20pt} \geq nP+\langle \vect{x}(i,T),\vect{s} + \vect{V} \rangle \; \mathrm{for\; some}\; j\neq i \; \mathrm{and}\; t'  \Big] \nonumber\\
& <K\exp(-n\delta_{1})
\end{align}
for all $\vect{s}'\in\mathcal{B}_n(\vect{s},\varepsilon)$.
\end{lemma}
\begin{proof}
For a particular $\vect{s}$, instead of \eqref{eq:Lem_QuantAdvrsryVec-Noisy-ErrEq}, 
let us assume that the code $\set{C}$ satisfies a stronger condition
\begin{align}\label{eq:Lem_QuantAdvrsryVec-Noisy-ModifiedErrEq}
&{\Pr}_{T} \Pr_V \Big[ \langle \vect{x}(j,t'),\vect{x}(i,T) + \vect{s} + \vect{V} \rangle \geq nP \nonumber\\
&\hspace{30pt}  - 2\varepsilon \sqrt{nP} +\langle \vect{x}(i,T),\vect{s}+\vect{V} \rangle \; \mathrm{for\; some}\; j\neq i \; \mathrm{and}\; t' \Big] \nonumber\\
& <K\exp(-n\delta_{1}). 
\end{align}
Then it can be verified that for all $\vect{s}'\in\mathcal{B}_n(\vect{s},\varepsilon)$
the code $\set{C}$ satisfies \eqref{eq:Lem_QuantAdvrsryVec-Noisy-ErrEq}
where $\vect{s}$ is replaced by $\vect{s}'$. To show this let $\vect{s}'=\vect{s}+\rho \vect{u}$
where $\vect{u}$ is an arbitrary unit vector and $\rho\in[-\varepsilon, \varepsilon]$.
Hence for all $\vect{s}'\in\mathcal{B}_n(\vect{s},\varepsilon)$ we can write
\begin{align*}
e(\vect{s}', i) &= {\Pr}_{T} \Pr_V \Big[ \langle \vect{x}(j,t'),\vect{x}(i,T)+\vect{s}'+\vect{V} \rangle \nonumber\\
&\hspace{17pt} \geq nP+\langle \vect{x}(i,T),\vect{s}' + \vect{V} \rangle \; \mathrm{for\; some}\; j\neq i \; \mathrm{and}\; t' \Big] \nonumber\\
&= {\Pr}_{T} \Pr_V \Big[ \langle \vect{x}(j,t'),\vect{x}(i,T)+\vect{s}+\vect{V} \rangle + \rho \langle \vect{x}(j,t'),\vect{u} \rangle \nonumber\\
&\hspace{17pt} \geq nP+\langle \vect{x}(i,T),\vect{s} + \vect{V} \rangle \nonumber\\
&\hspace{46pt} + \rho \langle \vect{x}(i,T), \vect{u} \rangle \; \mathrm{for\; some}\; j\neq i \; \mathrm{and}\; t' \Big] \nonumber\\
&\le {\Pr}_{T} \Pr_V \Big[ \langle \vect{x}(j,t'),\vect{x}(i,T)+\vect{s}+\vect{V} \rangle +  \varepsilon\sqrt{nP} \nonumber\\
&\hspace{17pt} \geq nP+\langle \vect{x}(i,T),\vect{s} + \vect{V} \rangle \nonumber\\
&\hspace{46pt} -\varepsilon\sqrt{nP} \; \mathrm{for\; some}\; j\neq i \; \mathrm{and}\; t' \Big] \nonumber\\
&\stackrel{\text{(a)}}{\le} K\exp(-n\delta_1),
\end{align*}
where (a) follows from \eqref{eq:Lem_QuantAdvrsryVec-Noisy-ModifiedErrEq}.

Now, in Lemma~\ref{lem:CodeProp_DoublyExpProb_Noisy} we can use
the stronger error requirement \eqref{eq:Lem_QuantAdvrsryVec-Noisy-ModifiedErrEq} 
to show that there exists a code which satisfies \eqref{eq:Lem_QuantAdvrsryVec-Noisy-ModifiedErrEq}.
This stronger requirement results in a rate loss, but 
as $\varepsilon$ goes to zero the rate loss due to that vanishes.
By the above argument, we know that this code satisfies \eqref{eq:Lem_QuantAdvrsryVec-Noisy-ErrEq}
for all $\vect{s}'\in\mathcal{B}_n(\vect{s},\varepsilon)$ and we are done.
\end{proof}

Finally, Lemma~\ref{lem:CodebookExistence-Noisy}
shows the existence of a good codebook for the quadratic constrained AVC
problem with stochastic encoder which have been introduced in Section~\ref{sec:ProbStatement-ScalarAVC}
and hence completes the proof of Theorem~\ref{thm:CapacityAVC-Noisy}.
\begin{lemma}[Codebook Existence] \label{lem:CodebookExistence-Noisy}
For every $\delta_{0}>\delta_{1}>0$ and $n\geq n_{0}(\delta_{0},\delta_{1})$ 
there exist a codebook $\set{C} = \{\vect{x}(i,t)\}$ of rate 
$R\leq\frac{1}{2} \log(1+\frac{P}{\sigma^2+\Lambda})$ comprises of 
vectors $\vect{x}(i, t) \in\mathbb{R}^{n}$ of size $\sqrt{nP}$ with $1\leq i\leq e^{nR}$ and 
$1\leq t\leq e^{n\delta_{0}}$ such that for every vector $\vect{s}$ and
every transmitted message $i$ we have
\begin{align}\label{eq:LemCodeProp-Noisy-ErrEq}
e(\vect{s}, i) &= {\Pr}_{T} \Pr_V \Big[ \langle \vect{x}(j,t'),\vect{x}(i,T)+\vect{s}+\vect{V} \rangle \nonumber\\
&\hspace{20pt} \geq nP+\langle \vect{x}(i,T),\vect{s} + \vect{V} \rangle \; \mathrm{for\; some}\; j\neq i \; \mathrm{and}\; t' \Big] \nonumber\\
& <K\exp(-n\delta_{1}). 
\end{align}
\end{lemma}
\begin{proof}
For any fixed codebook $\set{C}=\{\vect{x} (i,t)\}$, let us explicitly mention to 
the dependency of the error probability on $\set{C}$ by defining 
$e_{\set{C} } (\vect{s}, i) \triangleq e(\vect{s}, i)$. 
Then in order to prove the assertion of lemma we can equivalently show that
\[
\liminf_{n\rightarrow\infty} \Pr_{\set{C}^*}\left[ \forall\vect{s}, \forall i \:\:\: e_{\set{C}^*}(\vect{s}, i) < Ke^{-n\delta_1}\right] > 0.
\]
However, by using Lemma~\ref{lem:SmallPerturbation_s-NoisyAVC}, it is not necessary
to check for all $\vect{s}$ but only for those belonging to an $\varepsilon$-net\footnote{An $\varepsilon$-net 
is a set of points in a metric space such that each point of the 
space is within distance $\varepsilon$ of some point in the set.}
$\chi_n$ that covers $\mathcal{B}_n(0,\sqrt{n\Lambda})$. 

Hence, we can write
\begin{align*}
& \Pr_{\set{C}^*}\left[ \forall\vect{s}\in\chi_n, \forall i \:\:\: e_{\set{C}^*}(\vect{s}, i) < Ke^{-n\delta_1}\right] \nonumber\\
&\hspace{30pt}  = 1-\Pr_{\set{C}^*} \left[ \exists\vect{s}\in\chi_n,\exists i \:\:\: e_{\set{C}^*} (\vect{s}, i) \ge Ke^{-n\delta_1} \right] \nonumber\\
&\hspace*{30pt} \stackrel{\text{(a)}}{\ge} 1 - \sum_{\vect{s}\in \chi_n} \sum_{i=1}^{e^{nR}} \Pr_{\set{C}^*} \left[ e_{\set{C}^*} (\vect{s}, i) \ge Ke^{-n\delta_1} \right],
\end{align*}
where (a) follows from the union bound.

Now, note that to bound $|\chi_n|$ one might cover $\mathcal{B}_n(0,\sqrt{n\Lambda})$
by a hypercube of edge size $2\sqrt{n\Lambda}$; 
see Figure~\ref{fig:nu_dense_subset_covering}. So we can write
$|\chi_n| \le \left( \frac{2\sqrt{n\Lambda}}{\varepsilon} \right)^n$.
Then, by using Lemma~\ref{lem:CodeProp_DoublyExpProb_Noisy} we have
\begin{align*}
& \Pr_{\set{C}^*}\left[ \forall\vect{s}\in\chi_n, \forall i \:\:\: e_{\set{C}^*}(\vect{s}, i) < Ke^{-n\delta_1}\right] \nonumber\\
&\hspace{30pt} \ge 1- \left( \frac{2\sqrt{n\Lambda}}{\varepsilon} \right)^n  \times e^{nR}\times \exp\left( -K' e^{n(\delta_0-\delta_1)} \right),
\end{align*}
where, assuming $\delta_0 > \delta_1$, the right 
hand side goes to $1$ as $n$ goes to infinity and
this completes the proof of lemma.
\begin{figure}
\centering
\includegraphics[width=5.5cm]{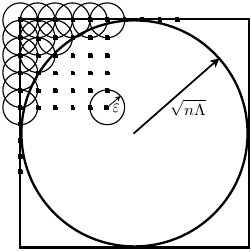}
\caption{This figure shows that how the whole sphere $\mathcal{B}(0,\sqrt{\Lambda})$ 
can be covered by $\varepsilon$-dense subsets $\chi_n$. Here the 
set $\chi_n$ comprises of points from a hyper-cubical lattice.}
\label{fig:nu_dense_subset_covering}
\end{figure}
\end{proof}

\subsection{Converse proof of Theorem~\ref{thm:CapacityAVC-Noisy}}\label{sec:AVC-Noisy-Proof-Converse}
The converse of Theorem~\ref{thm:CapacityAVC-Noisy} follows by 
combining two different upper bounds on the capacity.
The first bound follows by observing that if the
randomness of the stochastic encoder is also shared with the decoder
we can achieve higher rates. So by using result of \cite{NarHugh-IT87}
for randomized codes\footnote{Similar to our work, \cite{NarHugh-IT87}
also considers the maximum probability of error criterion.},
we have the following upper bound on the capacity of an AVC with
stochastic encoder
\[
C \le \frac{1}{2} \log\left(1+\frac{P}{\Lambda+\sigma^2} \right).
\]

Now, it only remains to show that $C=0$ for $P\le\Lambda$ where we
use a similar argument to \cite{BlBrTho-AnMatStat60} (also see \cite{CsisNar-IT91}).
To this end, we show that the adversary can fool the decoder and make it
confused. Because $P\le\Lambda$, the adversary can use a stochastic 
encoder $\Psi'$ with the same probabilistic characteristic of $\Psi$
where we assume that $\Psi$ and $\Psi'$ are independent\footnote{
Such a jamming strategy is equivalent to the notion of {\it symmetrizability} 
condition in the AVC literature (see, for instance~\cite{CsiN:88, CsiN:88a}, and \cite{LapNar-IT98}).}. 
Then for any decoder $\phi$ and for any $i\neq j$ we can write
\begin{align*}
& \Pr \left[ \phi \left( \Psi(i)+\Psi'(j) + \vect{V} \right) \neq i \right]  \\
&\hspace{75pt}	= \Pr \left[ \phi \left( \Psi(j)+\Psi'(i) + \vect{V} \right) \neq i \right] \\
&\hspace{75pt}	= 1-\Pr \left[ \phi \left( \Psi(j)+\Psi'(i) + \vect{V} \right) = i \right] \\
&\hspace{75pt}	\ge 1-\Pr \left[ \phi \left( \Psi(j)+\Psi'(i) + \vect{V} \right) \neq j \right].
\end{align*}
Hence we have
\begin{align*}
\frac{1}{M} \sum_{j=1}^{M} e_{\max}(\Psi'(j)) &\ge \frac{1}{M^2} \sum_{i,j=1}^{M}  e(\Psi'(j), i) \\
&= \frac{1}{M^2} \sum_{i,j=1}^{M} \Pr\left[ \phi \left( \Psi(i)+\Psi'(j) + \vect{V} \right) \neq i \right] \\
&\ge \frac{1}{M^2} \sum_{i,j=1}^{M} \Big[ \Pr\left[ \phi \left( \Psi(i)+\Psi'(j) + \vect{V} \right) \neq i \right] \\
& \quad + \Pr \left[ \phi \left( \Psi(j)+\Psi'(i) + \vect{V} \right) \neq j \right] \Big] \\
&\ge \frac{1}{M^2} \frac{M(M-1)}{2}\\
&\ge \frac{1}{4},
\end{align*}
where $M=e^{nR}$. This shows that
\[
\frac{1}{M} \sum_{j=1}^M \Ex\left[ e_{\max}(\Psi'(j)) \right] \ge \frac{1}{4},
\]
which means there exists at least a $k$ such that $\Ex\left[ e_{\max}(\Psi'(k)) \right]\ge \frac{1}{4}$
and this completes the proof.


\appendix

\begin{fact}\label{fact:ProbUpperBound}
For two events $\mathcal{A}$ and $\mathcal{B}$ we can write
\[
\Pr[\mathcal{A}] = \Pr[ \mathcal{A}\cap(\mathcal{B}\cup\bar{\mathcal{B}}) ]
	\le \Pr[ \mathcal{B} ] + \Pr[ \mathcal{A}\cap \bar{\mathcal{B}} ].
\]
\end{fact}

Our proof requires the following ``martingale concentration lemma'' 
proven in \cite[Lemma~A1]{CsisNar-IT91}.
\begin{lemma}[{\cite[Lemma~A1]{CsisNar-IT91}}]\label{lem:CsiszarNar-Lem-A1}
Let ${X}_{1},\ldots,{X}_{L}$
be arbitrary r.v.\textquoteright{}s and $f_{i}({X}_{1},\ldots,{X}_{L})$
be arbitrary function with $0\leq f_{i}\leq1$, $i=1,\ldots, L$. Then the
condition
\[
\Ex\left[ f_{i}({X}_{1},\ldots,{X}_{L}) | {X}_{1},\ldots, {X}_{i-1} \right]\leq a \:\:\:\: \mathrm{a.s.},\:\:\:\:\:\: i=1,\ldots, L,
\]
implies that
\[
\Pr \left[ \frac{1}{L}\sum_{i=1}^{L}f_{i}({X}_{1},\ldots, {X}_{i})>\tau \right] \leq \exp\left(-L(\tau\log2-a)\right).
\]
\end{lemma}

\begin{lemma}[{\cite[Lemma~2]{CsisNar-IT91}}] \label{lem:BoundDotProduct}
Let the random vector $\vect{U}$ be uniformly
distributed on the $n$-dimensional unit sphere. Then for every vector $\vect{u}$ on
this sphere and any $\frac{1}{\sqrt{2\pi n}}<\alpha<1$, we have 
\[
\Pr\left[ |\langle \vect{U}, \vect{u} \rangle|\geq\alpha \right] \leq 2(1-\alpha^{2})^{\frac{(n-1)}{2}}.
\]
\end{lemma}

\begin{proof}[Proof of Lemma~\ref{lem:CodeProp_DoublyExpProb_Noisy}]
For notational convenience let us normalize all vectors
$\vect{s}$, $\vect{V}$, and $\vect{X}(i,t)$ by $1/\sqrt{n}$ in this proof.

To derive the doubly exponential bound stated in the lemma,
we use Lemma~\ref{lem:CsiszarNar-Lem-A1}.
To this end let us define the functions $f_t$ for $1\le t\le e^{n\delta_0}$
as follows
\begin{align*}
& f_{t} \left( \vect{X}(i,1),\ldots,\vect{X}(i,t) \right) \\
&\hspace{15pt} \triangleq {\Pr}_{V} \Big[ \langle \vect{X}(j,t'),\vect{X}(i,t)+ \vect{s} +\vect{V} \rangle \\
&\hspace{24pt}  \geq P+\langle \vect{X}(i,t),\vect{s} + \vect{V} \rangle \: \text{for some $j\neq i$ and $t'$} \Big].
\end{align*}
Now, by using the functions $f_t$, the probability 
expression in the statement of lemma can be written as follows
\begin{align}\label{eq:LemDublyExp_Noisy_eq0}
&\Pr_{\set{C}^*} \bigg[ \Pr_T \Pr_{V} \Big[ \langle \vect{X}(j,t'),\vect{X}(i,T) +\vect{s}+\vect{V}\rangle \geq P \nonumber\\
&\quad\quad +\langle \vect{X}(i,T), \vect{s} +\vect{V} \rangle \; \mathrm{for}\; \mathrm{some}\; \mathrm{j}\neq \mathrm{i}\; \mathrm{and}\; t' \Big] \geq K e^{-n\delta_{1}} \bigg] \nonumber\\
&= \Pr_{\set{C}^*} \bigg[ \frac{1}{e^{n\delta_0}} \sum_{t=1}^{e^{n\delta_0}} \Pr_{V} \Big[ \langle \vect{X}(j,t'),\vect{X}(i,t) +\vect{s}+\vect{V}\rangle \geq P \nonumber\\
&\quad\quad +\langle \vect{X}(i,t), \vect{s} +\vect{V} \rangle \; \mathrm{for}\; \mathrm{some}\; \mathrm{j}\neq \mathrm{i}\; \mathrm{and}\; t'  \Big] \geq K e^{-n\delta_{1}} \bigg] \nonumber\\
&= \Pr_{\set{C}^*} \bigg[ \frac{1}{e^{n\delta_0}} \sum_{t=1}^{e^{n\delta_0}} f_t\left(\vect{X}(i,1),\ldots,\vect{X}(i,t) \right) \ge K e^{-n\delta_{1}} \bigg].
\end{align}
In order to bound \eqref{eq:LemDublyExp_Noisy_eq0} we use
Lemma~\ref{lem:CsiszarNar-Lem-A1}. To this end, we have to
bound the expected values of the functions $f_t$. So we 
proceed as follows
\begin{align}\label{eq:LemDublyExp_Noisy_eq1}
&\hspace{-5pt} \Ex_{\set{C}^*}  \left[ f_{t}( \vect{X}(i,1),\ldots,\vect{X}(i,t) ) | \vect{X}(i,1),\ldots,\vect{X}(i,t-1)  \right] \nonumber\\
&= \Ex_{\set{C}^*} \bigg[ \Pr_{V} \Big[ \langle \vect{X}(j,t'),\vect{X}(i,t) + \vect{s} + \vect{V}\rangle \nonumber\\
&\quad\geq P+\langle \vect{X}(i,t), \vect{s}\rangle +\langle \vect{X}(i,t), \vect{V}\rangle \nonumber\\
&\quad \: \text{for some $j\neq i$ and $t'$} \Big] \bigg| \vect{X}(i,1),\ldots,\vect{X}(i,t-1) \bigg] \nonumber\\
&\stackrel{\text{(a)}}{=} \Pr _{V} \bigg[ \Pr_{\set{C}^*} \Big[ \bigcup_{(j,t'):\: j\neq i} \big\{ \langle \vect{X}(j,t'),\vect{X}(i,t)+\vect{s}+\vect{V}\rangle  \nonumber\\
&\quad \geq P+\langle \vect{X}(i,t),\vect{s}+\vect{V}\rangle \big\} \Big] \bigg] \nonumber\\
&\stackrel{\text{(b)}}{\leq}  {\Pr}_{V}  \Pr_{\set{C}^*} \big[ \langle \vect{X}(i,t),\vect{s}+\vect{V} \rangle\leq -\delta_{2}  \big] \nonumber\\
&\quad + \Pr_{V} \bigg[ \Pr_{\set{C}^*} \Big[ \bigcup_{(j,t'): \: j\neq i} \big\{ \langle \vect{X}(j,t'),\vect{X}(i,t)+\vect{s}+\vect{V}\rangle \nonumber\\
&\quad \geq P +\langle \vect{X}(i,t),\vect{s}+\vect{V}\rangle \big\}, \langle \vect{X}(i,t),\vect{s}+\vect{V}\rangle > -\delta_{2} \Big] \bigg],
\end{align}
where (a) follows because $\vect{X}(i,t)$ are independent random 
variables so the conditioning can be removed and also
using the fact that for an event $\mathcal{A}$ we have 
$\Ex_{\set{C}^*}\Pr_V[\mathcal{A}]=\Pr_{\set{C}^*}\Pr_V[\mathcal{A}]$ and
(b) follows from Fact~\ref{fact:ProbUpperBound}.

Now, for $\delta_2>0$, by using Fact~\ref{fact:ProbUpperBound} we can bound
the first term of \eqref{eq:LemDublyExp_Noisy_eq1} as follows
\begin{align}\label{eq:LemDublyExp_Noisy_eq2}
&\Pr_{V}  \Pr_{\set{C}^*} \big[ \langle \vect{X}(i,t),\vect{s}+\vect{V}\rangle\leq{-\delta_{2}} \big] \nonumber\\
&\leq \Pr_{V} \big[ \| \vect{s}+\vect{V} \|^2 \geq \| \vect{s} \|^2 +\sigma^2+\delta_{2} \big] \nonumber\\
&\quad + \Pr_{V} \Pr_{\set{C}^*} \big[ \langle \vect{X}(i,t),\vect{s}+\vect{V}\rangle \le -\delta_{2}, \nonumber\\
&\hspace{60pt} \| \vect{s}+\vect{V} \|^2 < \| \vect{s} \|^2 +\sigma^2+\delta_{2} \big] \nonumber\\
&\leq \Pr_{V} \big[ \| \vect{s}+\vect{V} \|^2 \geq \| \vect{s} \|^2 +\sigma^2+\delta_{2} \big] \nonumber\\
&\quad + \Pr_{V} \Pr_{\set{C}^*} \big[ |\langle \vect{X}(i,t),\vect{s}+\vect{V}\rangle| \geq \delta_{2}, \nonumber\\
&\hspace{60pt} \| \vect{s}+\vect{V} \|^2 < \| \vect{s} \|^2 +\sigma^2+\delta_{2} \big].
\end{align}
First note that
$\| \vect{s}+\vect{V} \|^{2} ={\| \vect{s} \|}^2+{ \|\vect{V} \|}^2+2\langle\vect{s},\vect{V}\rangle$.
Then, since $\vect{V}=(\vect{V}_{1},\ldots,\vect{V}_{n})$ is a sequence of i.i.d. Gaussian random
variables $N(0,\frac{\sigma^{2}}{n})$, the first term of
\eqref{eq:LemDublyExp_Noisy_eq2} can be bounded as follows
\begin{align}\label{eq:LemDublyExp_Noisy_eq3}
\Pr_V \big[ & \|\vect{s}+\vect{V} \|^{2}  >\|\vect{s}\|^2 + \sigma^{2}+\delta_{2} \big] \nonumber\\
&=\Pr \big[ \|\vect{V} \|^{2}+2\langle\vect{s},\vect{V}\rangle > \sigma^{2}+\delta_{2} \big] \nonumber\\
&\stackrel{\text{(a)}}{\leq} \Pr [ \langle\vect{s},\vect{V}\rangle\geq \eta ]  +\Pr [ \|\vect{V}\|^{2}+2\eta > \sigma^{2}+\delta_{2} ] \nonumber\\
&\stackrel{\text{(b)}}{=}\Pr \left[ \langle\vect{u},\vect{V}\rangle\geq \frac{\eta}{ \| \vect{s} \|} \right] +\Pr \big[ \|\vect{V}\|^2 > \sigma^{2}+\delta_{2} - 2\eta \big],
\end{align}
where (a) follows from Fact~\ref{fact:ProbUpperBound} for $\eta>0$ and in (b)
we define $\vect{u}=\vect{s}/\| \vect{s} \|$.
Because  $\vect{u}$ is a unitary vector 
it is straightforward to show that $\langle \vect{u},\vect{V} \rangle \sim N(0,\frac{\sigma^2}{n})$.
Hence the first term in \eqref{eq:LemDublyExp_Noisy_eq3} can be
bounded as follows
\begin{equation}\label{eq:LemDublyExp_Noisy_eq4}
\Pr_{V} \left[ \langle \vect{s},\vect{V} \rangle \ge \eta \right] =Q\left (\frac{\sqrt{n}\eta}{\sigma \|\vect{s}\| } \right)
	\le \frac{1}{2} \exp{ \left( - \frac{\eta^2 n}{2\sigma^2 \Lambda}  \right)},
\end{equation}
where in the above equation we have used the approximation $Q(y)\leq\frac{1}{2}e^{-\frac{y^{2}}{2}}$.
In order to bound the second term in \eqref{eq:LemDublyExp_Noisy_eq3}
note that $\frac{n}{\sigma^2} \| \vect{V} \|^2$ has the Chi-squared distribution with $n$ degree of freedom.
Then by using \cite[Lemma~1]{LaurentMassart-AnnStat00} we can bound 
the second term of \eqref{eq:LemDublyExp_Noisy_eq3} as follows
\begin{align}\label{eq:LemDublyExp_Noisy_eq5}
&\Pr_V \left[ \frac{n}{\sigma^2} \| \vect{V} \|^{2}> \left( 1+\frac{\delta_{2}-2\eta}{\sigma^2} \right) n \right] \nonumber\\
&\hspace{20pt} \le \exp\left( -\frac{1}{2}\left[ 1+\frac{\delta_2-2\eta}{\sigma^2}-\sqrt{1+2\frac{\delta_2-2\eta}{\sigma^2}} \right] n \right) \nonumber\\
&\hspace{20pt} = \exp(-\xi n),
\end{align}
where $\xi=\frac{1}{2}\left[ 1+\frac{\delta_2-2\eta}{\sigma^2}-\sqrt{1+2\frac{\delta_2-2\eta}{\sigma^2}} \right]$ is a positive quantity if $\delta_2>2\eta$.

\begin{remark}\label{rmrk:LemDublyExp_Noisy_Bound_xi}
Note that because $\sqrt{1+2x}\leq 1+x-\frac{x^2}{8}$ for 
all $x\in[0,1]$ then by choosing $x=\frac{\delta_2-2\eta}{\sigma^2}$ 
we have $\xi\geq \frac{1}{16}{(\frac{\delta_2-2\eta}{\sigma^2})}^{2}$ 
and $\exp(-n\xi)\leq \exp(-\frac{1}{16}{(\frac{\delta_2-2\eta}{\sigma^2})}^{2})$. 
\end{remark}

Now it remains to bound the second term of \eqref{eq:LemDublyExp_Noisy_eq2}. 
To this end let us write
\begin{align*}
&{\Pr}_{V}{\Pr}_{\set{C}^*} \left[ |\langle\vect{X}(i,t),\vect{s}+\vect{V}\rangle| \geq \delta_{2},\|\vect{s}+\vect{V} \|^2 < \| \vect{s} \|^2 + \sigma^2 + \delta_{2} \right] \\
& =\int_{0}^{\| \vect{s}\|^2 + \sigma^{2}+\delta_{2}}{\Pr}_{V}{\Pr}_{\set{C}^*} \Big[ |\langle \vect{X}(i,t),\vect{s}+\vect{V}\rangle| > \delta_{2} \: \Big|   \\
&\hspace{160pt}  \|\vect{s}+\vect{V}\|^{2}=r \Big] dF(r) 
\end{align*}
where $F(r)=\Pr \left[ \|\vect{s}+\vect{V}\|^{2} \leq r \right]$.
Then we can write
\begin{align*}
&{\Pr}_{V}{\Pr}_{\set{C}^*} \left[ |\langle\vect{X}(i,t),\vect{s}+\vect{V}\rangle| \geq \delta_{2},\|\vect{s}+\vect{V} \|^2 < \| \vect{s} \|^2 + \sigma^2 + \delta_{2} \right] \\
&=\int_{0}^{\| \vect{s} \|^2 + \sigma^{2} +\delta_{3}} \Pr _{V} \Pr _{\set{C}^*} \Big[ \langle \vect{X}(i,t),\frac{\vect{s}+\vect{V}}{\|\vect{s}+\vect{V}\|}\rangle> \frac{\delta_{2}}{\sqrt{r}} \Big| \\
&\hspace{160pt} \| \vect{s} + \vect{V} \|^{2} = r \Big] dF(r)\\
& \stackrel{\text{(a)}}{\leq} {\Pr}_{U}{\Pr}_{\Breve{\set{C}}^*} \left[ \langle \Breve{\vect{X}}(i,t), \vect{U} \rangle > \frac{\delta_{2}/\sqrt{P}}{ \sqrt{ \|\vect{s}\|^2 + \sigma^{2}+\delta_{2}}} \right]
\end{align*}
where $\vect{U}=\frac{\vect{s}+\vect{V}}{\|\vect{s}+\vect{V}\|}$,
$\Breve{\vect{X}}(i,t)=\frac{\vect{X}(i,t)}{\|\vect{X}(i,t)\|}$, 
and (a) is true because evaluating the term inside
the integration for the point $r=\| \vect{s} \|^2 + \sigma^{2} +\delta_{2}$
can only increase the probability term. Next,
it follows that
\begin{align}\label{eq:LemDublyExp_Noisy_eq6}
&\Pr_{V}{\Pr}_{\set{C}^*} \left[ |\langle\vect{X}(i,t),\vect{s}+\vect{V}\rangle| \geq \delta_{2},\|\vect{s}+\vect{V} \|^2 < \| \vect{s} \|^2 + \sigma^2 + \delta_{2} \right] \nonumber\\
&=\int \Pr_{\Breve{\set{C}}^*} \left[\langle \Breve{\vect{X}}(i,t),\vect{u}\rangle > \frac{\delta_{2}/\sqrt{P}}{\sqrt{ \|\vect{s}\|^2 + \sigma^{2}+\delta_{2}}} \Big| \vect{U}=\vect{u} \right] f_{\vect{U}}(\vect{u})d\vect{u} \nonumber\\
%
&\stackrel{\text{(a)}}{\leq} \int 2\left(1-\frac{{{\delta}_{2}}^{2} /P}{\|\vect{s}\|^2+\sigma^{2}+{\delta}_{2}} \right)^{\frac{n-1}{2}}f_{\vect{U}}(\vect{u})d\vect{u} \nonumber\\
&= 2\left(1-\frac{{{\delta}_{2}}^{2}/P}{\|\vect{s}\|^2+\sigma^{2}+{\delta}_{2}} \right)^{\frac{n-1}{2}} \nonumber\\
&\stackrel{\text{(b)}}{\leq} 2\exp\left(-\frac{n-1}{2}\frac{{{\delta}_{2}}^{2} /P}{\|\vect{s}\|^2+\sigma^{2}+{\delta}_{2}} \right)
\end{align}
where (a) follows from Lemma~\ref{lem:BoundDotProduct} and
(b) follows from the inequality $1-x\le e^{-x}$ for $0<x<1$.
Finally, by combining \eqref{eq:LemDublyExp_Noisy_eq2}, \eqref{eq:LemDublyExp_Noisy_eq3},
\eqref{eq:LemDublyExp_Noisy_eq4}, \eqref{eq:LemDublyExp_Noisy_eq5}, and
\eqref{eq:LemDublyExp_Noisy_eq6} we can bound the 
first term in \eqref{eq:LemDublyExp_Noisy_eq1} as follows
\begin{align}\label{eq:LemDublyExp_Noisy_Expct_fk_Bound_1}
&\Pr_{V} \Pr_{\set{C}^*} \left[ \langle \vect{X}(i,t),\vect{s}+\vect{V}\rangle\leq{-\delta_{2}} \right] \leq \nonumber\\
&\quad 2\exp\left(-\frac{n-1}{2}\frac{{{\delta}_{2}}^{2}/P}{\|\vect{s}\|^2+\sigma^{2}+{\delta}_{2}} \right)+e^{-n\xi}+\frac{1}{2}e^{-\frac{n{\eta}^2}{2{\sigma}^2\Lambda}}. 
\end{align}

Now we bound the second term in \eqref{eq:LemDublyExp_Noisy_eq1} as follows.
Suppose $\mathcal{A}$ denotes for the event $\{ \langle \vect{X}(i,t),\vect{s}+\vect{V}\rangle > -\delta_{2} \}$
and let $\phi=\langle \vect{X}(i,t),\vect{s}+\vect{V}\rangle$.
Then for the second term of \eqref{eq:LemDublyExp_Noisy_eq1},
we note that
\begin{align}
&\Pr_{V\set{C}^*} \bigg[ \bigcup_{(j,t'): \: j\neq i} \big\{ \langle \vect{X}(j,t'),\vect{X}(i,t)+\vect{s}+\vect{V}\rangle  \geq P + \phi \big\}, \mathcal{A}  \bigg] \nonumber\\
&\stackrel{(a)}{\leq}  \Pr_{V} \left[ \|\vect{s}+\vect{V}\|^2 \ge \|\vect{s}\|^2 + \sigma^2+\delta_2 \right] \nonumber\\
&+ \Pr_{V\set{C}^*} \bigg[ \bigcup_{\begin{subarray}{c} (j,t'): \\ j\neq i \end{subarray}} \big\{ \langle \vect{X}(j,t'),\vect{X}(i,t)+\vect{s}+\vect{V}\rangle  \geq P + \phi \big\}, \mathcal{A}, \mathcal{B}  \bigg] \nonumber\\
&\stackrel{(b)}{\leq}  e^{-n\xi} + \frac{1}{2}e^{-\frac{n\eta^2}{2\sigma^2\Lambda}} \nonumber\\
& + \sum_{(j,t'): \: j\neq i } \Pr_{V\set{C}^*} \left[\langle \vect{X}(j,t'),\vect{X}(i,t)+\vect{s}+\vect{V}\rangle \geq P + \phi, \mathcal{A}, \mathcal{B} \right], \nonumber
\end{align}
where (a) follows from Fact~\ref{fact:ProbUpperBound} and
we use $\mathcal{B}$ to denote the event $\{ \|\vect{s}+\vect{V}\|^2 < \|\vect{s}\|^2 + \sigma^2+\delta_2 \}$.
The first two terms in (b) follow from \eqref{eq:LemDublyExp_Noisy_eq3},
\eqref{eq:LemDublyExp_Noisy_eq4}, and \eqref{eq:LemDublyExp_Noisy_eq5}
while the third term is a result of the union bound.
Let us define the unit vectors $\breve{\vect{X}}(j,t')=\frac{\vect{X}(j,t')}{\|\vect{X}(j,t')\|}$
and $\vect{U}=\frac{\vect{X}(i,t)+\vect{s}+\vect{V}}{ \|\vect{X}(i,t)+\vect{s}+\vect{V}\| }$.
Then we note that
\begin{align}
&\Pr_{V\set{C}^*} \bigg[ \bigcup_{(j,t'): \: j\neq i} \big\{ \langle \vect{X}(j,t'),\vect{X}(i,t)+\vect{s}+\vect{V}\rangle  \geq P + \phi \big\}, \mathcal{A}  \bigg] \nonumber\\
&\stackrel{\text{(a)}}{\le} e^{-n\xi} + \frac{1}{2}e^{-\frac{n\eta^2}{2\sigma^2\Lambda}} \nonumber\\
& \quad + \sum_{(j,t'): \: j\neq i } \Pr_{U\Breve{\set{C}}^*} \bigg[ \langle \breve{\vect{X}}(j,t'), \vect{U} \rangle \nonumber\\
&\hspace{85pt} \ge \frac{P+\phi}{\sqrt{P} \sqrt{P+\|\vect{s}+\vect{V}\|^2 +2\phi} } \Big| \mathcal{A}, \mathcal{B} \bigg] \nonumber\\
&\stackrel{\text{(b)}}{\le} e^{-n\xi} + \frac{1}{2}e^{-\frac{n\eta^2}{2\sigma^2\Lambda}} \nonumber\\
& \quad + \sum_{(j,t'): \: j\neq i } \Pr_{U\Breve{\set{C}}^*} \bigg[ \langle \breve{\vect{X}}(j,t'), \vect{U} \rangle \nonumber\\
&\hspace{85pt} \ge \frac{P-\delta_2}{\sqrt{P} \sqrt{P+\Lambda +\sigma^2+\delta_2 -2\delta_2} }  \bigg] \nonumber
\end{align}
where in (a) we use the fact that 
$\Pr[\mathcal{E},\mathcal{A},\mathcal{B}] \le \Pr[\mathcal{E}|\mathcal{A},\mathcal{B}]$
and (b) follows because by substituting $\|\vect{s}+\vect{V}\|^2=\Lambda+\sigma^2+\delta_2$
and $\phi=-\delta_2$ the probability term in front of 
the summation in (a) can only increase; this implies that we can 
remove the conditioning with respect to events $\mathcal{A}$ 
and $\mathcal{B}$.
Now, by applying Lemma~\ref{lem:BoundDotProduct}, we can
further bound the second term of \eqref{eq:LemDublyExp_Noisy_eq1}
as follows
\begin{align}\label{eq:LemDublyExp_Noisy_Expct_fk_Bound_2}
&\Pr_{V\set{C}^*} \bigg[ \bigcup_{(j,t'): \: j\neq i} \big\{ \langle \vect{X}(j,t'),\vect{X}(i,t)+\vect{s}+\vect{V}\rangle  \geq P + \phi \big\}, \mathcal{A}  \bigg] \nonumber\\
& \le e^{-n\xi} + \frac{1}{2}e^{-\frac{n\eta^2}{2\sigma^2\Lambda}} \nonumber\\
&\quad + 2e^{n(R+\delta_0)} \left(1-\frac{P-\delta'_2}{P+\Lambda +\sigma^2 -\delta_2} \right)^{\frac{n-1}{2}} \nonumber\\
& \le e^{-n\xi} + \frac{1}{2}e^{-\frac{n\eta^2}{2\sigma^2\Lambda}} \nonumber\\
&\quad + 2e^{n(R+\delta_0) + \frac{n-1}{2} \log{ \left(1-\frac{P-\delta'_2}{P+\Lambda +\sigma^2 -\delta_2} \right) }},
\end{align}
where $\delta'_2=2\sqrt{P}\delta_2-\delta_2^2$.

Finally, by combining \eqref{eq:LemDublyExp_Noisy_Expct_fk_Bound_1}
and \eqref{eq:LemDublyExp_Noisy_Expct_fk_Bound_2} we can write
the following bound for the expectation of functions $f_t$
\begin{align*}
&\Ex_{\set{C}^*} \left[ f_{t}( \vect{X}(i,1),\ldots,\vect{X}(i,t) ) | \vect{X}(i,1),\ldots,\vect{X}(i,t-1)  \right] \nonumber\\
&\le 2\exp\left(-\frac{n-1}{2}\frac{{{\delta}_{2}}^{2}/P}{\|\vect{s}\|^2+\sigma^{2}+{\delta}_{2}} \right)+ 2e^{-n\xi} + e^{-\frac{n{\eta}^2}{2{\sigma}^2\Lambda}} \nonumber\\
& \quad+ 2e^{n(R+\delta_0) + \frac{n-1}{2} \log{ \left(1-\frac{P-\delta'_2}{P+\Lambda +\sigma^2 -\delta_2} \right) }}.
\end{align*}
By making some more assumptions on $\delta_0$, $\delta_2$,
$\eta$, $R$, and introducing $\delta_1$, 
we can simplify the upper bounds on the expected values of
functions $f_t$ as follows
\begin{align*}
&\Ex_{\set{C}^*} \left[ f_{t}( \vect{X}(i,1),\ldots,\vect{X}(i,t) ) | \vect{X}(i,1),\ldots,\vect{X}(i,t-1)  \right] \nonumber\\
&\stackrel{\text{(a)}}{\leq} 2\exp{\left(-n\frac{1}{16}{(\frac{\delta_{2}-2\eta}{\sigma^2})}^2 \right)} +\exp{\left(-\frac{n\eta^2}{2\sigma^2\Lambda}\right)}\\
&\quad +2\exp{\left(-(\frac{n-1}{2})\frac{\delta_{2}^2/P}{2{({\|\vect{s}\|}^2+\sigma ^2)}}\right)} +2\exp{(-n\delta_{1})}\\
&\stackrel{\text{(b)}}{\leq} 2\exp{(-n\delta_{1})} + \exp{(-n\delta_1)} + 2\exp{(-n\delta_1)} +2\exp{(-n\delta_{1})}\\
&\le 10\exp{(-n\delta_1)}
\end{align*}
where (a) follows by Remark~\ref{rmrk:LemDublyExp_Noisy_Bound_xi}, 
assuming $\delta_{2}\leq {\|\vect{s}\|}^2+\sigma^2$, 
and choosing
\[
R < \frac{1-1/n}{2}\log\left(1+\frac{P-\delta'_2}{\Lambda+\sigma^{2}-\delta_2+\delta'_2}\right)-\delta_0-\delta_{1},
\]
(b) follows by assuming the conditions $\delta_{2}>2\eta+4\sigma^2\sqrt{\delta_{1}}$, $\eta>\sqrt{2\Lambda\sigma^2 \delta_1}$, and
$\delta_2>\sqrt{\frac{4P(\Lambda+\sigma^2)\delta_1}{1-1/n}}$.

Then by applying Lemma~\ref{lem:CsiszarNar-Lem-A1} and 
choosing $a=10 e^{-n\delta_1}$ and $\tau=Ke^{-n\delta_1}$ 
we have
\begin{align*}
& \Pr_{\set{C}^*} \bigg[ \frac{1}{e^{n\delta_0}} \sum_{t=1}^{e^{n\delta_0}} \Pr_{V} \big[ \langle \vect{X}(j,t'),\vect{X}(i,t) +\vect{s}+\vect{V}\rangle \geq P \nonumber\\
&\quad\quad +\langle \vect{X}(i,t), \vect{s} +\vect{V} \rangle \; \mathrm{for}\; \mathrm{some}\; \mathrm{j}\neq \mathrm{i}\; \mathrm{and}\; t' \big] \Big] \geq K e^{-n\delta_{1}} \bigg] \nonumber\\
&\le \exp\left(-\exp(n\delta_0) \Big( K\log{2}\exp(-n\delta_1)-10\exp(-n\delta_1) \Big) \right) \nonumber\\
&= \exp\Big( -(K\log{2}-10) \exp(n(\delta_0-\delta_1)) \Big).
\end{align*}
By assuming $\delta_0>\delta_1>0$ we obtain the desired
doubly exponential bound, hence we are done.
\end{proof}

{\scriptsize{
\bibliographystyle{IEEEtran}
\bibliography{gauss_obl}
}}

%
%
%
%

\end{document}